%% file: main-tree-alpha.tex
\documentclass[11pt, a4paper]{article}
\usepackage[margin=1in]{geometry}
\input{commands}

\title{Finding large sparse induced subgraphs in graphs\\
of small (but not very small) tree-independence number}

\author{%
Daniel Lokshtanov%
\thanks{University of California, Santa Barbara.}%
\and%
Michał Pilipczuk%
\thanks{University of Warsaw. Supported by the European Research Council (ERC) under the European Union’s Horizon 2020 research and innovation programme, grant agreement no 948057~(BOBR).}%
\and %
Pawe{\l} Rz\k{a}\.zewski%
\thanks{Warsaw University of Technology \& University of Warsaw. Supported by the National Science Centre grant number 2024/54/E/ST6/00094.}
}

\date{}

\begin{document}

\maketitle

\begin{abstract}
    The independence number of a tree decomposition is the size of a largest independent set contained in a single bag.
    The tree-independence number of a graph $G$ is the minimum independence number of a tree decomposition of $G$.
    As shown recently by Lima et al. [ESA~2024], a large family of optimization problems asking for a maximum-weight induced subgraph of bounded treewidth, satisfying a given \textsf{CMSO}$_2$ property, can be solved in polynomial time in graphs whose tree-independence number is bounded by some constant~$k$.

    However, the complexity of the algorithm of Lima et al. grows rapidly with $k$, making it useless if the tree-independence number is superconstant. In this paper we present a refined version of the algorithm. We show that the same family of problems can be solved in time~$n^{\mathcal{O}(k)}$, where $n$ is the number of vertices of the instance, $k$ is the tree-independence number, and the $\Oh(\cdot)$-notation hides factors depending on the treewidth bound of the solution and the considered \textsf{CMSO}$_2$ property.

    This running time is quasipolynomial for classes of graphs with polylogarithmic tree-independence number; several such classes were recently discovered. Furthermore, the running time is subexponential for many natural classes of geometric intersection graphs --  namely, ones that admit balanced clique-based separators of sublinear size.
\end{abstract}

\section{Introduction}

Tree decompositions and treewidth are among most successful concepts in graph theory, with numerous structural~\cite{DBLP:journals/jgt/HarveyW17,DBLP:journals/jct/RobertsonS86,DBLP:conf/stoc/ChekuriC13,DBLP:journals/actaC/Bodlaender93} and algorithmic applications~\cite{DBLP:journals/actaC/Bodlaender93,DBLP:journals/dam/ArnborgP89,DBLP:journals/jacm/BodlaenderFLPST16} (see \cref{sec:treewidth} for formal definitions).
In the context of the latter, having a tree decomposition of a graph with small (say, constant) width allows us to mimic the standard bottom-up dynamic programming on trees in order to solve many computational problems efficiently.
This idea led to a flurry of algorithmic results and techniques, designed for graphs of bounded treewidth~\cite{DBLP:journals/talg/CyganNPPRW22,DBLP:journals/talg/FockeMINSSW25,DBLP:journals/siamcomp/BasteST23,DBLP:conf/soda/FominLS14,journals/talg/BodlaenderCNPPS16}.

In order to understand which problems can be solved efficiently on given class of graphs,
it is beneficial to design algorithms that do not focus on particular problems, but rather on \emph{families of problems} expressible in a certain way. Such results are typically referred to as \emph{algorithmic meta-theorems}. On a high level, an algorithmic meta-theorem says that every problem of a given type can be solved efficiently in the given class of graphs.
Such ``types of problems'' are usually defined in terms of logic.

Probably the best known algorithmic meta-theorem is the celebrated result of Courcelle~\cite{DBLP:journals/iandc/Courcelle90} which says that every problem expressible in \cmsotwo can be solved in linear time in graphs given with a tree decomposition of constant width. 
Here, \cmsotwo is the {\em{Monadic Second Order}} logic of graphs with quantification over edge subsets and modular predicates. In essence, the logic allows quantification over single vertices and edges as well as over subsets of vertices and of edges.
See \cref{sec:cmso} for more information on \cmsotwo and the result of Courcelle.  

While Courcelle's theorem handles a very broad family of problems, it comes at a cost: the assumption that the treewidth is bounded by a constant is rather restrictive.
In particular, the theorem applies only to sparse graphs. On the other hand, there are classes of dense graphs with very good structural properties, like interval graphs~\cite[Chapter 8]{Golumbic2004}, circular-arc graphs~\cite[Chapter 8.6]{Golumbic2004}, or chordal graphs~\cite[Chapter 4]{Golumbic2004}. Often these properties can be expressed in terms of the existence of a certain tree decomposition whose every bag has some ``nice structure'' (opposed to every bag having bounded size, as in the definition of treewidth).
For instance, chordal graphs are precisely the graphs that admit a tree decomposition where every bag is a clique.
Such decompositions with a ``nice structure'' can be also exploited~algorithmically.

This motivation led to the definition of \emph{tree-independence number} of a graph.
This parameter was independently introduced by Yolov~\cite{DBLP:conf/soda/Yolov18} and by Dallard et al.~\cite{dallard2022firstpaper}; we use the notation and terminology from the latter work.
The independence number of a tree decomposition is the minimum value of $k$ such that every bag induces a subgraph of independence number at most~$k$. The tree-independence number of a graph $G$ is the minimum independence number of a tree decomposition of $G$; see \cref{sec:treewidth} for more information. Thus, chordal graphs are exactly graphs of tree-independence number 1.

The algorithmic motivation behind the definition of tree-independence number was to understand the tractability of the \textsc{Max (Weight) Independent Set} problem (\MWIS). Indeed, a standard dynamic programming approach shows that \MWIS can be solved in polynomial time for graphs of bounded tree-independence number, given with an appropriate tree decomposition.
However, it turns out that many other problems, asking for finding large induced subgraphs that are sparse (or, equivalently in this context, of bounded clique number), can be efficiently solved in the considered classes~\cite{DBLP:conf/soda/Yolov18}.
In particular, many of these problems can be expressed as special cases of a certain abstract problem.
For an integer $t$ and a \cmsotwo sentence $\psi$, we define $(\tw <t,\psi)$-\MWIS (here, \MWIS stands for \emph{maximum-weight induced subgraph}) as follows.

\problemTask{$(\tw <t,\psi)$-\MWIS}%
{A graph $G$ equipped with a weight function $\wei\colon V(G) \to \Q_+$.}%
{Find a set $F \subseteq V(G)$, such that
\begin{itemize}
\item $G[F] \models \psi$,
\item $\tw(G[F]) <t$, and
\item $F$ is of maximum weight subject to the conditions above,
\end{itemize}
or conclude that no such set exists.}

Notable special cases of the $(\tw <t,\psi)$-\MWIS problem are \MWIS, \textsc{Feedback Vertex Set} (equivalently, \textsc{Max Induced Forest}), and \textsc{Even Cycle Transversal}  (equivalently, \textsc{Max Induced Odd Cactus}).  

As shown by Lima et al.~\cite{DBLP:conf/esa/LimaMMORS24}, every problem expressible as $(\tw <t,\psi)$-\MWIS can be efficiently solved in graphs of bounded tree-independence number (given with a decomposition).

\begin{theorem}[Lima et al.~\cite{DBLP:conf/esa/LimaMMORS24}]\label{thm:cmso-old}
For any fixed integers $t,k$ and a fixed \cmsotwo sentence $\psi$, there exists $d$ such that the following holds:
The $(\tw <t,\psi)$-\MWIS problem can be solved in time $\Oh(n^d)$ in $n$-vertex graphs of tree-independence number at most~$k$.
\end{theorem}

Let us remark that the assumption that the solution induces a sparse graph is crucial here.
In particular, \textsc{Clique} is \textsf{NP}-hard in graphs of tree-independence number 2,
as \MWIS is \textsf{NP}-hard in triangle-free graphs~\cite{Poljak}.
Interestingly, algorithms for $(\tw <t,\psi)$-\MWIS are known also for other classes of well-structured dense graphs~\cite{MR4752925,DBLP:conf/stoc/GartlandLPPR21,DBLP:conf/soda/ChudnovskyMPPR24,FominTV15}.

\cref{thm:cmso-old} allows us to solve a rich family of optimization problems in classes of constant tree-independence number~\cite{dallard2022secondpaper,DBLP:journals/jgt/AbrishamiACHSV24a}.
However, the fact that the exponent $d$ depends (at least) exponentially on $k$ makes the algorithm from \cref{thm:cmso-old} not very useful for classes where tree-independence is small, but superconstant -- say, logarithmic in $n$.
It turns out that there are some natural classes where such phenomenon can be observed; we discuss them in \cref{sec:applications}.

In this paper, we prove the following refinement of \cref{thm:cmso-old}.

\begin{theorem}\label{thm:cmsotech}
Let $t$ be a fixed integer and $\psi$ a fixed \cmsotwo sentence.
Then the $(\tw <t,\psi)$-\MWIS problem can be solved in time $n^{\Oh(k)}$ in $n$-vertex graphs given along with a tree decomposition of independence number at most $k$ and $n^{\Oh(1)}$ nodes.
\end{theorem}

The key difference, compared to \cref{thm:cmso-old}, is that the dependence of running time of the algorithm in \cref{thm:cmsotech} on $k$ is single-exponential. This makes \cref{thm:cmsotech} useful to obtain quasipolynomial-time or subexponential-time algorithms in classes where the tree-independence number is, respectively, polylogarithmic or sublinear.

Let us remark that the requirement that the input graph is given along with a suitable tree decompostion is not very restrictive.
Indeed, Dallard et al.~\cite{dallard2022computing} proved that given an $n$-vertex graph with tree-independence number at most $k$, in time  $2^{\Oh(k^2)}n^{\Oh(k)}$ one can obtain a tree decomposition with independence number at most $8k$.
By combining this result with our \cref{thm:cmsotech}, we obtain the following result.

\begin{theorem}\label{thm:cmain}
Let $t$ be a fixed integer and $\psi$ be a fixed \cmsotwo sentence.
Then the $(\tw <t,\psi)$-\MWIS problem can be solved in time $2^{\Oh(k^2)}n^{\Oh(k)}$ in $n$-vertex graphs of tree-independence number at most~$k$.
\end{theorem}

We remark that if $k = \Omega(\sqrt{n})$, then the $2^{\Oh(k^2)}$ factor in the running time of the algorithm from \cref{thm:cmain} makes it somewhat useless. Thus, in such cases we need to rely on obtaining an appropriate tree decomposition in some other way.

\subsection{Algorithmic consequences.}\label{sec:applications}
The main motivation behind improving \cref{thm:cmso-old} to \cref{thm:cmsotech} is to use it for classes with superconstant tree-independence number.
Let us discuss some examples of such classes.

\paragraph{Classes of polylogarithmic tree-independence number.}
Let us start with a few definitions. A \emph{hole} in a graph is an induced cycle with at least 4 vertices. A hole is even is it has an even number of vertices.

A \emph{theta} is a graph consisting of two distinct vertices $a, b$ and three paths $P_1, P_2, P_3$ from $a$ to
$b$ such that the union of any two of them induces a hole.
A \emph{pyramid} is a graph consisting of a vertex $a$ and a triangle $b_1,b_2,b_3$, and three paths $P_i$ from $a$ to $b_i$ for $i\in \{1,2,3\}$  such that the union of any two of these paths induces a hole.
Finally, a \emph{generalized prism} is a graph consisting of two (not necessarily disjoint) triangles $a_1,a_2,a_3$ and $b_1,b_2,b_3$, and three
paths $P_i$ from $a_i$ to $b_i$ for $i\in \{1,2,3\}$, such that the union of any two of these paths induces a hole.
A \emph{three-path configuration}, or a 3PC for short, is a graph that is either a theta, a pyramid, or a generalized~prism.

By $S_{q,q,q}$ we denote the graph obtained from three paths, each with $q+1$ vertices, by picking one endvertex of each path and identifying them.
By $K_{q,q}$ we denote the complete graph with each side of size $q$.
The $q \times q$-wall, denoted by $W_{q \times q}$ is the $q \times q$ hexagonal grid.
Then $\mathcal{L}_q$ is the family of graphs such that every $G \in \mathcal{L}_q$ is the line graph of some subdivision of $W_{q \times q}$.

For a family $\mathcal{H}$ of graphs, we say that $G$ is {\em{$\mathcal{H}$-free}} if it does not contain any graph from $\mathcal{H}$ as an induced subgraph.
A graph is 3PC-free (resp., even-hole-free, theta-free, pyramid-free) if none of its induced subgraphs is a 3PC (resp., even hole, theta, pyramid).
To improve the readability, we use parenthesis when we forbid several families simultaneously; e.g., 3PC-free graphs are precisely (theta, pyramid, generalized prism)-free graphs.

The next theorem summarizes known results about families of graphs where the tree-independence number is bounded by a polylogarithmic function of the number of vertices.

\begin{theorem}[Chudnovsky et al.~\cite{DBLP:journals/jctb/ChudnovskyHLS26}, Chudnovsky et al.~\cite{DBLP:conf/soda/ChudnovskyGHLS25}, Chudnovsky et al.~\cite{DBLP:journals/corr/abs-2501-14658}, Chudnovsky et al.~\cite{DBLP:journals/corr/abs-2509-15458}]\label{thm:bounds-for-treealpha}
    Let $q$ be a fixed constant.
    Let $G$ be an $n$-vertex graph for $n \geq 2$. The tree-independence number of $G$ is bounded by 
    \begin{itemize}
    \item $\Oh(\log^2 n)$ if $G$ is 3PC-free~\cite{DBLP:journals/jctb/ChudnovskyHLS26};
    \item $\Oh(\log^{10}n)$ if $G$ is even-hole-free~\cite{DBLP:conf/soda/ChudnovskyGHLS25};
    \item $\Oh(\log^4 n)$ if $G$ is ($S_{q,q,q}$, $K_{q,q}$; $\mathcal{L}_q$)-free~\cite{DBLP:journals/corr/abs-2501-14658}; and    
    \item $\Oh(\log^{c_q}n)$ if $G$ is (theta, pyramid, $\mathcal{L}_q$)-free, where $c_q$ is some constant depending on $q$~\cite{DBLP:journals/corr/abs-2509-15458}. \end{itemize}
\end{theorem}

Observe that \cref{thm:cmain} implies that the $(\tw <t,\psi)$-\MWIS problem can be solved in quasipolynomial time in every class of graphs where the tree-independence number is bounded by a polylogaritmic function of the number of vertices. Combining this observation with \cref{thm:bounds-for-treealpha}, we obtain the following corollary.

\begin{corollary}
    For any fixed $t$ and a \cmsotwo sentence $\psi$, the $(\tw <t,\psi)$-\MWIS problem can be solved in quasipolynomial time for
    \begin{itemize}
    \item 3PC-free graphs;
    \item even-hole-free graphs;
    \item ($S_{q,q,q}$, $K_{q,q}$, $\mathcal{L}_q$)-free graphs, for any fixed $q$; and 
    \item (theta, pyramid, $\mathcal{L}_q$)-free graphs, for any fixed $q$.           
    \end{itemize}
\end{corollary}

\paragraph{Clique-based separators in geometric intersection graphs.}
The next family of examples comes from intersection graphs of certain geometric objects.
For a graph $G$ and a real number $\beta \in (0,1)$, a \emph{$\beta$-balanced clique-based separator} is a family $\mathcal{S}$ of subsets of $V(G)$, each inducing a clique in $G$, so that every component of $G - \bigcup_{S \in \mathcal{S}}$ has at most $\beta \cdot |V(G)|$ vertices. The size of a clique-based separator $\mathcal{S}$ is $|\mathcal{S}|$, i.e., the number of cliques.\footnote{An astute reader might notice that typically in clique-based separators we do not optimize size, but weight, defined as $\sum_{S \in \mathcal{S}} \log (|S| +1)$. However, in the context of our work, size is a more appropriate measure. Furthermore, these two measures are very close: the size is bounded by the weight, which is in turn bounded by the size times the logarithm of the number of vertices.}
We say that a class $\mathcal{G}$ of graphs \emph{admits balanced clique-based separators of size $f(n)$}, for some function $f$, if there is an absolute constant $\beta <1$ such that every $n$-vertex graph in $\mathcal{G}$ has a $\beta$-balanced clique-based separator of size $f(n)$.

It turns out that many natural classes of geometric intersection graphs have balanced separators of sublinear size. 
We list some of the results below, for more, we refer to the work of de Berg et al.~\cite{DBLP:journals/siamcomp/BergBKMZ20}, de Berg et al.~\cite{DBLP:journals/algorithmica/BergKMT23}, and Dallard et al.~\cite{DBLP:journals/corr/abs-2506-12424}. There you can also find precise definitions of classes listed in \cref{thm:bounds-for-cliquebased}.

\begin{theorem}[de Berg et al.~\cite{DBLP:journals/siamcomp/BergBKMZ20}, de Berg et al.~\cite{DBLP:journals/algorithmica/BergKMT23}, Dallard et al.~\cite{DBLP:journals/corr/abs-2506-12424}]\label{thm:bounds-for-cliquebased}
    The class $\mathcal{G}$ admits balanced clique-based separators of size $\Oh(n^\delta)$, where:
    \begin{itemize}
        \item $\mathcal{G}$ is the class of intersection graphs of convex fat objects  in $\mathbb{R}^2$ and $\delta = 1/2$~\cite{DBLP:journals/siamcomp/BergBKMZ20};
        \item $\mathcal{G}$ is the class of map graphs and $\delta = 2/3$~\cite{DBLP:journals/algorithmica/BergKMT23};
        \item $\mathcal{G}$ is the class of pseudodisk intersection graphs and $\delta = 2/3$~\cite{DBLP:journals/algorithmica/BergKMT23}; and
        \item $\mathcal{G}$ is the class of contact segment graphs and $\delta = 2/3$~\cite{DBLP:journals/corr/abs-2506-12424}.
    \end{itemize}
    Furthermore, if the graph is given with a geometric representation, then the separator can be found in time polynomial in the size of the representation.
\end{theorem}

The existence of balanced clique-based separators with sublinear size allows us to solve many computational problems, like \MWIS or  \textsc{Feedback Vertex Set}, in \emph{subexponential} time~\cite{DBLP:journals/siamcomp/BergBKMZ20}. However, to the best of our knowledge, no metatheorem is known in this setting.

Let us argue that our \cref{thm:cmsotech} can be applied here.
First, it is not hard to observe that when combined with the usual approach of building tree decompositions from balanced separators,
the existence of sublinear balanced clique-based separators yields tree decompositions of sublinear independence number (see \cite[Theorem 20]{DBLP:journals/tcs/Bodlaender98} or \cite[Lemma 2]{DBLP:conf/wads/AnCO23} or \cite[Lemma 1.1]{DBLP:journals/corr/abs-2506-12424}).

\begin{lemma}\label{lem:from-sep-to-td}
Let $\mathcal{G}$ be a hereditary class that admits balanced clique-based separators of size $\Oh(n^\delta)$ for some $\delta <1$.
Then every $n$-vertex graph from $\mathcal{G}$ has tree-independence number $\Oh(n^\delta)$.

Furthermore, if balanced separators can be found in polynomial time, then the required tree decomposition can also be constructed in polynomial time.
\end{lemma}

Now, combining \cref{thm:bounds-for-cliquebased}, \cref{lem:from-sep-to-td}, and our \cref{thm:cmain}, we obtain the following.

\begin{theorem}\label{thm:cliquebased}
    For any fixed $t$ and a \cmsotwo sentence $\psi$, the $(\tw <t,\psi)$-\MWIS problem can be solved in time
    \begin{itemize}
        \item $2^{\Oh(\sqrt{n}\log n)}$ in the class of intersection graphs of convex fat objects in $\mathbb{R}^2$; and
        \item $2^{\Oh(n^{2/3}\log n)}$ in the class of map graphs, of pseudodisk intersection graphs, or of contact segment graphs.
    \end{itemize}
\end{theorem}

\section{Basic notation and tools}

For a nonnegative integer $p$, we denote $[p]\coloneqq \{1,\ldots,p\}$. For a set $A$, by $\id_A$ we denote the identity function on $A$.

We use standard graph notation and terminology. All graphs considered in this paper are finite, simple, and undirected; that is, a graph $G$ consists of a finite vertex set $V(G)$ and an edge set $E(G)\subseteq \binom{V(G)}{2}$.
We also consider vertex-weighted graphs, i.e., graphs $G$ equipped with a function $\wei \colon V(G) \to \mathbb{Q}_+$.
We assume that arithmetic operations on weights can be performed in constant time.
For a set of $X$ vertices, we define $\wei(X) = \sum_{v \in X}\wei(v)$.

For a graph $G$, by $\alpha(G)$ we denote the number of vertices in a largest independent set in $G$.
All subgraphs in this paper are induced, and we often identify such subgraphs with their vertex sets.
For example, for a set $X \subseteq V(G)$, we write $\alpha(X)$ as a shorthand for $\alpha(G[X])$.

\subsection{Tree decompositions and tree-independence}\label{sec:treewidth}

A tree decomposition of a graph $G$ is a pair $\cT = (T, \{X_u\}_{u \in V(T)})$, where $T$ is a tree and,
for every $u \in V(T)$, $X_u$ is a subset of $V(G)$ called a \emph{bag}, so that the following conditions are met:
\begin{enumerate}[(i)]
    \item for every vertex $v \in V(G)$ there is $u \in V(T)$ such that $v \in X_u$;
    \item for every edge $vw \in E(G)$ there is $u \in V(T)$ such that $w,v \in X_u$; and
    \item for every vertex $v \in V(G)$, the set of nodes $u \in V(T)$ such that $v \in X_u$ induces a connected subgraph of $T$.
\end{enumerate}
To avoid confusion between $V(G)$ and $V(T)$, the elements of $V(T)$ will be called {\em{nodes}}.
Typically, we consider $T$ as a rooted tree, which defines the ancestor-descendant relation on its nodes.
This relation is reflexive, i.e., we consider every node its own ancestor as well as its own descendant.
For a node $u \in V(T)$, we define $V_u \coloneqq \bigcup \{X_{u'}\colon u' \text{ is a descendant of }u\}$.

The \emph{width} of a tree decomposition $\cT$ is defined as $\max_{u \in V(T)} |X_u|-1$.
The \emph{treewidth} of a graph $G$, denoted by $\tw(G)$, is the minimum width of a tree decomposition of $G$.

While computing treewidth is \textsf{NP}-hard~\cite{ArnborgCorneilProskurowski:1987}, we can efficiently verify whether $\tw(G) <t$,  for any fixed $t$.
\begin{theorem}[Bodlaender~\cite{DBLP:journals/siamcomp/Bodlaender96}]\label{thm:treewidthfpt}
    Let $t$ be a fixed constant.
    Given an $n$-vertex graph $G$, in time $\Oh(n)$ one can compute a tree decomposition of $G$ of width at most $t$, or decide that $\tw(G) > t$.
\end{theorem}


\paragraph{Tree-independence number.}
For a tree decomposition $\cT = (T, \{X_u\}_{u \in V(T)})$, by $\alpha(\cT)$ we denote $\max_{u \in V(T)} \alpha(X_u)$.
By \emph{tree-independence number} of a graph $G$ we denote the minimum value of $\alpha(\cT)$ over all tree decompositions $\cT$ of $G$.

\paragraph{Nice tree decompositions.}
A tree decomposition  $\cT = (T, \{X_u\}_{u \in V(T)})$ rooted at $r$ is \emph{nice} if $X_r = \emptyset$ and every node $u$ of $T$ is of one of the following types:
\begin{description}
    \item[Leaf node.]  $u$ is a leaf of $T$ and $X_u = \emptyset$,
    \item[Introduce node.] $u$ has exactly one child $u'$, and $X_u = X_{u'} \cup \{v\}$ for some $v \notin X_{u'}$, 
    \item[Forget node.] $u$ has exactly one child $u'$, and $X_u = X_{u'} \setminus \{v\}$ for some $v \in X_{u'}$,
    \item[Join node.] $u$ has exactly two children $u',u''$ and $X_u = X_{u'} = X_{u''}$.
\end{description}

It is well-known (and easy to verify) that, given any tree decomposition $\cT$, in polynomial time one can obtain a nice tree decomposition  $\cT'$ whose every bag is a subset of some bag of $\cT$.
In particular, $\textsf{width}(\cT') \leq \textsf{width}(\cT)$ and $\alpha(\cT') \leq \alpha(\cT)$.

\subsection{\cmsotwo }\label{sec:cmso}

We use the standard logic \cmsotwo, which stands for Counting Monadic Second Order logic with edge quantification, on graphs. We now recall the fundamentals and establish the notation; for a broader introduction see e.g.~\cite{DBLP:books/daglib/0030804}.

In \cmsotwo on graphs, there are variables of four different sorts: for single vertices and for single edges (typically denoted with lower case letters), as well as for subsets of vertices and for subsets of edges (typically denoted with upper case letters). Atomic formulas of \cmsotwo are of the following form:
\begin{itemize}
	\item equality of variables of the same sort: $x=y$ or $X=Y$;
	\item testing the incidence relation: $\mathsf{inc}(x,e)$, where $x$ is a vertex variable and $e$ is an edge variable;
	\item set membership: $x\in X$, where $x$ is a (vertex/edge) variable and $X$ is a (vertex/edge) set variable; and
	\item modular arithmetic checks: $|X|\equiv a\bmod m$, where $X$ is a (vertex/edge) set variable and $a,m\in \Z$ are fixed integers with $m>0$.
\end{itemize}
Starting with those atomic formulas, larger formulas can be constructed using standard boolean connectives ($\wedge,\vee,\Rightarrow$), negation, and quantification over variables of all sorts (both universal and existential). This defines the syntax of \cmsotwo. The semantics, i.e., how the satisfaction of a formula in a graph is defined, is as expected. For $p\in \N$, by \cpmsotwo we denote the fragment of \cmsotwo where all the moduli $m$ used in the modular arithmetic checks satisfy~$m\leq p$. The {\em{quantifier rank}} of a \cmsotwo formula $\phi$ is the maximum nesting depth of quantifiers in $\phi$.

In a formula of \cmsotwo, any variable that is not bound by a quantifier is called a {\em{free variable}}. We often distinguish the free variables in parentheses by the formula, e.g. $\phi(x,y)$. A {\em{sentence}} is a formula without free variables. As usual, we write $G\models \phi(\tup a)$ to state that formula $\phi$ holds in a graph $G$ for an evaluation $\tup a$ of the free variables of $\phi$.

We will use the following standard fact asserting that the class of graphs of treewidth at most $t$ is definable in \cmsotwo, for every fixed $t$.

\begin{lemma}\label{lem:twt}
    For every $t$ there is a sentence  $\phi_{\tw <t}$ of \cmsotwo such that for every graph $G$, \[G \models \phi_{\tw <t}\qquad\textrm{if and only if}\qquad \tw(G) <t.\]
\end{lemma}
\begin{proof}
     By the celebrated result of Robertson and Seymour~\cite{DBLP:journals/jct/RobertsonS04}, for each $t$ there is a \emph{finite} family $\mathcal{F}_t$ of graphs such that, for every $G$ it holds that $\tw(G) <t$ if and only if $G$ does not contain any $F \in \mathcal{F}_t$ as a minor.
     Thus, the sentence $\phi_{\tw <t}$ is a conjunction of sentences asserting that a graph does not contain $F$ as a minor, over all $F \in \mathcal{F}_t$.
     Expressing that a fixed graph $F$ is  a minor of a given graph can be easily done in \cmsotwo.  Such a sentence should say that there exists a vertex set for each vertex of $F$ (i.e., a constant number of sets) so that each of these sets induces a connected graph, the sets are pairwise disjoint, and an edge between a pair of sets exists whenever the two corresponding vertices of $H$ are adjacent.
\end{proof}

\subsection{\cmsotwo and boundaried graphs}

\paragraph{Boundaried graphs.} For a finite set of vertex variables $\tup x$, a graph $G$, and a vertex subset $A\subseteq V(G)$, by $A^{\tup x}$ we denote the set of all injective functions from $\tup x$ to $A$; and $G^{\tup x}$ is a shorthand for $V(G)^{\tup x}$. An element $\tup a\in A^{\tup x}$ is interpreted as an evaluation mapping the variables of $\tup x$ to pairwise different vertices of $A$: each $x\in \tup x$ is evaluated to $\tup a(x)$. The set $\tup a(\tup x)$, that is, the image of $\tup a$, is denoted by $\im \tup a$. An {\em{$\tup x$-boundaried graph}} is a pair $(G,\tup a)$, where $G$ is a graph and $\tup a\in G^{\tup x}$. The {\em{boundary}} of~$(G,\tup a)$ is the image of $\tup a$.

We define the following two operations on boundaried graphs:
\begin{itemize}
	\item Suppose $(G_1,\tup a_1)$ and $(G_2,\tup a_2)$ are an $\tup x_1$-boundaried graph and an $\tup x_2$-boundaried graph, respectively. Then {\em{gluing}} those two boundaried graphs results in the $\tup x$-boundaried graph $(G,\tup a)=(G_1,\tup a_1)\oplus(G_2,\tup a_2)$, where $\tup x=\tup x_1\cup \tup x_2$, constructed as follows. Take the disjoint union of $G_1$ and $G_2$ and for each $x\in \tup x_1\cap \tup x_2$, fuse $\tup a_1(x)$ with $\tup a_2(x)$ into one vertex and declare it to be $\tup a(x)$. For $x\in \tup x_1\setminus \tup x_2$ set $\tup a(x)\coloneqq \tup a_1(x)$ and for $x\in \tup x_2\setminus \tup x_1$ set $\tup a(x)\coloneqq \tup a_2(x)$.
	\item If $(G,\tup a)$ is an $\tup x$-boundaried graph and $x\in \tup x$ is a variable, then $\forget_x(G,\tup a)$ is the $(\tup x\setminus \{x\})$-boundaried graph $(G,\tup a|_{\tup x\setminus \{x\}})$. In other words, the graph stays the same, but we forget  variable $x$ from the domain $\tup x$. 
\end{itemize}

\paragraph{Types.} It is well-known (see e.g.~\cite{DBLP:journals/algorithmica/BoriePT92}) that for any fixed $p,q\in \N$ and a finite set of vertex variables $\tup x$, every formula $\phi(\tup x)$ of \cpmsotwo of quantifier rank at most $q$ can be effectively transformed into an equivalent normalized formula $\phi'(\tup x)$ of \cpmsotwo of the same quantifier rank. Furthermore, the set of the normalized formulas, call it $\Formulae^{\tup x,p,q}$, is finite and computable from $p,q,\tup x$. We define $\msotypes^{\tup x,p,q}$ to be the powerset of $\Formulae^{\tup x,p,q}$, and for an $\tup x$-boundaried graph $(G,\tup a)$, we define its {\em{rank-$(p,q)$ type}} as
\[\msotype^{p,q}(G,\tup a)\coloneqq \{\phi(\tup x)\in \Formulae^{\tup x,p,q}~|~G\models \phi(\tup a)\} \in \msotypes^{\tup x,p,q}.\]
Two $\tup x$-boundaried graphs $(G_1,\tup a_1)$ and $(G_2,\tup a_2)$ are called {\em{$(p,q)$-equivalent}}, denoted $(G_1,\tup a_1)\equiv_{p,q} (G_2,\tup a_2)$, if they have the same rank-$(p,q)$ types, that is, $\msotype^{p,q}(G_1,\tup a_1)=\msotype^{p,q}(G_2,\tup a_2)$.

Note that isomorphic boundaried graphs are $(p,q)$-equivalent for all $p,q\in \N$, where an isomorphism between $\tup x$-boundaried graphs $(G_1,\tup a_1)$ and $(G_2,\tup a_2)$ is defined as an isomorphism between $G_1$ and $G_2$ that maps $\tup a_1(x)$ to $\tup a_2(x)$, for each $x\in \tup x$. Further, for any fixed $p$ and $q$, the rank-$(p,q)$ type of a boundaried graph $(G,\tup a)$ is computable from $(G,\tup a)$. Finally, it is well known that types defined in this are compositional with respect to the operations on boundaried graphs, in the sense that the type of the result of an operation is determined by the type(s) of the input graph(s). This is formalized in the following lemma.

\begin{lemma}[see e.g.~{\cite[Proposition~8]{DBLP:conf/stoc/GartlandLPPR21}}]\label{prop:mso-type}
Fix $p,q\in \N$. Then:
\begin{itemize}
\item For every pair of finite sets of vertex variables $\tup x_1,\tup x_2$ there exists a computable binary operation $\oplus\colon \msotypes^{\tup x_1,p,q}\times \msotypes^{\tup x_2,p,q}\to \msotypes^{\tup x,p,q}$, where $\tup x=\tup x_1\cup \tup x_2$, so that for any $\tup x_1$-boundaried graph $(G_1,\tup a_1)$ and any $\tup x_2$-boundaried graph $(G_2,\tup a_2)$, we have
\[\msotype^{p,q}((G_1,\tup a_1)\oplus (G_2,\tup a_2))=\msotype^{p,q}(G_1,\tup a_1)\oplus \msotype^{p,q}(G_2,\tup a_2).\]
Moreover, the operation $\oplus$ is associative and commutative.
\item For every finite set of vertex variables $\tup x$ and a variable $x\in \tup x$ there exists a computable operation $\forget_{x}\colon \msotypes^{\tup x,p,q}\to \msotypes^{\tup y,p,q}$, where $\tup y\coloneqq \tup x\setminus \{x\}$, so that for any $\tup x$-boundaried graph $(G,\tup a)$, we have
\[\msotype^{p,q}(\forget_x(G,\tup a))=\forget_x(\msotype^{p,q}(G,\tup a)).\]
\end{itemize}
\end{lemma}

For a sentence $\psi$ of \cpmsotwo of quantifier rank at most $q$ and a finite set of vertex variables~$\tup x$, we let $\msotypes^{\tup x,p,q}[\psi]\subseteq \msotypes^{\tup x,p,q}$ be the set of those rank-$(p,q)$ types that contain the normalized sentence equivalent to $\psi$. Thus, for any $\tup x$-boundaried graph $(G,\tup a)$, we have 
\[G\models \psi\qquad\textrm{if and only if}\qquad \msotype^{p,q}(G,\tup a)\in \msotypes^{\tup x,p,q}[\psi].\]
Note that the set $\msotypes^{\tup x,p,q}[\psi]$ is computable from $p,q,\tup x,\phi$.

Finally, it will be often convenient to speak about types of graphs with a boundary understood as a prescribed subset of vertices, say $B$, without explicitly enumerating the elements of $B$ using an explicit set of variables. To this end, for a graph $G$ and a set of vertices $B\subseteq V(G)$, we write
\[\msotype^{p,q}(G,B)\coloneqq \msotype^{p,q}(G,\id_B),\]
where on the right hand side we treat $B$ as a finite set of variables and $\id_B\in B^B$ is an evaluation that maps each of these variables to itself. Consequently, if $G,G'$ are two graphs and $B\subseteq V(G)\cap V(G')$, then we write $(G,B)\equiv_{p,q}(G',B)$ if and only if  $(G,\id_B)\equiv_{p,q}(G',\id_B)$.

\paragraph*{\cmsotwo on graphs of bounded treewidth.} We will use the classic result of Courcelle about the tractability of \cmsotwo on graphs of bounded treewidth, in the following formulation.

\begin{theorem}[Courcelle~\cite{DBLP:journals/iandc/Courcelle90}]\label{thm:courcelleb}
	Fix a constant $\ell\in \N$, a finite set of vertex variables $\tup x$, and a formula of $\phi(\tup x)$ of \cmsotwo. Then given an $n$-vertex boundaried graph $(G,\tup a)$ of treewidth at most~$\ell$, one can decide whether $\phi(\tup a)$ holds in $G$ in time $\Oh(n)$.
\end{theorem}

We note that the original statement of Courcelle assumed that a tree decomposition of width at most $\ell$ is also given on input, but this can be computed in $\Oh(n)$ time using \cref{thm:treewidthfpt}.




From \cref{thm:courcelleb} we may immediately derive that one can efficiently compute the type of a given graph of bounded treewidth.

\begin{corollary}\label{cor:type-computation}
	Fix constants $p,q,\ell\in \N$ and a finite set of vertex variables $\tup x$.
	Then given an $n$-vertex $\tup x$-boundaried graph $(G,\tup a)$ of treewidth at most $\ell$, one can compute $\msotype^{\tup x,p,q}(G,\tup a)$ in time $\Oh(n)$.
\end{corollary}
\begin{proof}
	Apply \cref{thm:courcelleb} for every formula $\phi(\tup x)\in \Formulae^{\tup x,p,q}$, noting that the set $\Formulae^{\tup x,p,q}$ is finite and computable from $\tup x,p,q$.
\end{proof}

\section{Signatures}
In what follows we assume that $p,q,t$, and $\ell$ are fixed nonnegative integers. We treat $p,q,t$ as constants, in particular we do not attempt to optimize any dependence on them and hide functions depending on $p,q,t$ in the $\Oh(\cdot)$ notation.
On the other hand, we treat $\ell$ as a parameter and make the dependence on it explicit in the notation; we try to keep this dependence as low as possible.

We fix a set $\tup x$ consisting of at most $\ell$ vertex variables. Consider an $\tup x$-boundaried graph $(G,\tup a)$ of treewidth smaller than $t$. Let us emphasize that while the treewidth of $G$ is considered constant (namely smaller than $t$), the size of the boundary of $(G,\tup a)$ is not, as it depends on~$\ell$.
In this section we show that in such a setting, the number of possible $(p,q)$-types of $(G,\tup a)$ is bounded by a moderate function of $\ell$. The key idea towards bounding this number is to understand the $(p,q)$-type of $(G,\tup a)$ through the lens of the {\em{signature}} of $(G,\tup a)$, defined below.


\paragraph{Defining signatures.} Consider any rooted tree decomposition $\cT = (T, \{X_u\}_{u \in V(T)})$ of $G$ of width smaller than $t$. Without loss of generality assume that $T$ is binary, that is, every node of $T$ has at most $2$ children. Let $B\coloneqq \im \tup a\subseteq V(G)$ be the image of $\tup a$.
For each vertex $v \in B$, mark one, arbitrarily chosen, node of $T$ whose bag contains $v$. Note that we have marked at most $\ell$ nodes of $T$.

Now, let $Q$ be the lowest common ancestor closure of all the marked nodes. By this we mean the set containing all the marked nodes and, for every pair of marked nodes, also their their (unique) lowest common ancestor in $T$.
It is well known (see e.g.~\cite{DBLP:journals/jacm/BodlaenderFLPST16}) that (i) $|Q|<2\ell$, and (ii) every connected component of $T-Q$ is adjacent to at most two different vertices of $Q$. Define \[B'\coloneqq \bigcup_{u \in Q} X_u.\] 
Clearly, we have $|B'| \leq t \cdot 2\ell = \Oh(t\ell)$.

Let $s$ be the number of components of $T-Q$, and let $T_1,\ldots,T_s$ be those components. Since $T$ is binary, we have $s \leq 2\ell$.
For each $i \in [s]$, let $Z_i$ be the set of at most two nodes from $Q$ that are adjacent to $V(T_i)$.
Now, for each $i \in [s]$, we define \[P_i \coloneqq G\left[\bigcup_{u \in V(T_i) \cup Z_i} X_u\right]\qquad\textrm{and}\qquad B_i \coloneqq \bigcup_{u \in Z_i} X_u.\]
Note that we have $|B_i|\leq 2t$.

We fix a set $\Omega$ consisting of $2\ell t$ vertex variables, disjoint from $\tup x$. ($\Omega$ is fixed independent of~$(G,\tup a)$, so that for all the considered boundaried graphs $(G,\tup a)$ we use the same $\Omega$.)
To every vertex $v\in B'\setminus B$, we assign a different variable $y_v\in \Omega$. Let $\tup y\coloneqq \tup x\cup \{y_v\colon v\in B'\setminus B\}\subseteq \tup x\cup \Omega$ and define $\tup b\in G^{\tup y}$ by setting $\tup b(x)=\tup a(x)$ and $\tup b(y_v)=v$ for each $v\in B'\setminus B$. Thus, $(G,\tup b)$ is a $\tup y$-boundaried graph with boundary $B'$. Finally, for each $i\in [s]$, we define $\tup y_i\coloneqq \tup b^{-1}(B_i)$ and $\tup b_i\coloneqq \tup b|_{\tup y_i}$. Thus, $(P_i,\tup b_i)$ is a $\tup y_i$-boundaried graph with boundary $B_i$ of size at most $2t$.

Now comes the crucial definition. Given the objects described above, we define a \emph{signature} of $(G,\tup a)$ as a tuple consisting of:
\begin{itemize}
	\item the number $s$;
	\item the variable set $\tup y\subseteq \tup x\cup \Omega$ and its subsets $\tup y_1,\ldots,\tup y_s$;
	\item the graph $H$ on the vertex set $\tup y$ where $y,y'\in \tup y$ are adjacent if and only if $\tup b(y)$ and $\tup b(y')$ are adjacent in $G$; and
	\item for each $i\in [s]$, the type $\tau_i\coloneqq \msotype^{p,q}(P_i,\tup b_i)$.
\end{itemize}
Note that a single boundaried graph $(G,\tup a)$ may have multiple signatures, depending on the choice of the tree decomposition $\cT$, the marked nodes, the enumeration of the connected components of $T-Q$, etc. However, some signature of $(G,\tup a)$ can be computed efficiently.

\begin{lemma}\label{lem:compute-signature}
	For fixed $p,q,t\in \N$ there is an algorithm that given a $\tup x$-boundaried graph $(G,\tup a)$, some signature of $(G,\tup a)$ can be computed in polynomial time.
\end{lemma}
\begin{proof}
	We start by calling \cref{thm:treewidthfpt} to compute a tree decomposition $\cT$ of $G$ of width less than $t$.
	Computing $B'$, variable sets $\tup y, \tup y_1,\ldots,\tup y_s$, evaluations $\tup b,\tup b_1,\ldots,\tup b_2$, and boundaried graphs $(P_i,\tup b_i)$, for all $i\in [s]$, can be done in polynomial time right from the definition.
	Since each $P_i$ is an induced subgraph of $G$, the treewidth of $P_i$ is smaller than $t$ as well. Therefore, we may use \cref{cor:type-computation} to compute each type $\tau_i=\msotype^{p,q}(P_i,\tup b_i)$ in time linear in $|V(P_i)|$.    
\end{proof}

Next, we bound the number of possible distinct  signatures.

\begin{lemma}\label{lem:number-signature}
	There are $\ell^{\Oh(\ell)}$ possible distinct signatures of $\tup x$-boundaried graphs of treewidth less than $t$. 
\end{lemma}
\begin{proof}
	We bound the number of ways of choosing the consecutive components of a signature:
	\begin{itemize}
		\item There are at most $2\ell+1$ ways to choose $s$.
		\item As $|\Omega|=2\ell t=\Oh(\ell)$, there are $2^{\Oh(\ell)}$ ways to choose $\tup y$.
		\item As each set $\tup y_i$ has size at most $2t$, there are at most $(|\Omega|^{2t})^s\leq \ell^{\Oh(\ell)}$ ways to choose~$\tup y_1,\ldots,\tup y_s$.
		\item Note that $H$ is isomorphic to an induced subgraph of $G$, hence its treewidth is smaller than $t$. It is well-known that a graph of treewidth less than $t$ has fewer than $t\ell=\Oh(\ell)$ edges. As $H$ has $\Oh(\ell)$ vertices, there are $\ell^{\Oh(\ell)}$ ways to choose $H$.
		\item For each $i\in [s]$, the number of ways to choose the type $\tau_i$ is bounded by $|\msotypes^{\tup y_i,p,q}|$, which is $\Oh(1)$ due to $|\tup y_i|\leq 2t$.
	\end{itemize}
	The product of the upper bounds presented above is $\ell^{\Oh(\ell)}$, and from the description it is straightforward to enumerate $\ell^{\Oh(\ell)}$ candidates for signatures.
\end{proof}


Finally, we prove the crucial property: equality of signatures implies equality of types.

\begin{lemma}\label{lem:equiv-signature}
	Let $(G_1,\tup a_1)$ and $(G_2,\tup a_2)$ be two $\tup x$-boundaried graphs 
	of treewidth less than $t$. Suppose there is a signature $\sigma$ that is both a signature of $(G_1,\tup a_1)$ and a signature of $(G_2,\tup a_2)$.
    Then $(G_1,\tup a_1) \equiv_{p,q} (G_2,\tup a_2)$.
\end{lemma}
\begin{proof}
	Let $\sigma=(s,\tup y,\{\tup y_i\colon i\in [s]]\},H,\{\tau_i\colon i\in [s]\})$. We note that a graph isomorphic to $(G_1,\tup a_1)$ can be obtained from $(H,\id_{\tup y})$, by iteratively gluing a $\tup y_i$-boundaried graph of rank-$(p,q)$ type $\tau_i$, for each $i\in [s]$, and then iteratively applying the $\forget_y$ operation for all the variables $y\in \tup y\setminus \tup x$, in any order. And similarly, $(G_2,\tup a_2)$ can be also obtained in this way. By \cref{prop:mso-type}, the rank-$(p,q)$ type of the result of such an operation is uniquely determined by the types $\{\tau_i\colon i\in [s]\}$ and the type $\msotype^{\tup y,p,q}(H,\id_{\tup y})$. It follows that $(G_1,\tup a_1)$ and $(G_2,\tup a_2)$ have the same rank-$(p,q)$~types.
\end{proof}


Some remarks are in place.
Recall that $|\tup x|$ can be as large as $\ell$, which is not a constant. 
Consequently, the number of formulae in $\msotypes^{\tup x,p,q}$ can be very large (essentially, a $q$-fold exponential function of $\ell$) and in particular, we cannot efficiently compute $\msotype^{\tup x,p,q}(G,\tup a)$, for a given $\tup x$-boundaried graph $(G,\tup a)$. However, if we additionally know that the treewidth of $G$ is bounded by a constant, then using \cref{lem:compute-signature} we may efficiently compute a signature of $(G,\tup a)$, for which, by \cref{lem:number-signature}, there are only $\ell^{\Oh(\ell)}$ many possibilities. And by \cref{lem:equiv-signature}, if for two $\tup x$-boundaried graphs we compute the same signature, then they must have the same rank-$(p,q)$ type. Note that the reverse implication is not necessarily true: it is possible that for two $(p,q)$-equivalent graphs we compute two different signatures.



We now use the above observations to prove the following lemma, which encapsulates the outcome of this section in a single statement. In essence, it says that given a large family of boundaried graphs, we can efficiently compress it by still keeping at least one member of each equivalence class of $(p,q)$-equivalence.

\begin{lemma}\label{lem:compress}
    Let $t,p,q\in \N$ be fixed constants and let $\ell$ be an integer parameter. 
    Let $G$ be an $n$-vertex graph with vertex weights $\wei \colon V(G) \to \Q_+$, and $B\subseteq V(G)$ be a set of at most $\ell$ vertices of $G$.
    Further, let $\mathcal{F}$ be a family of subsets of $V(G)$ such that for every $F \in \mathcal{F}$, it holds that $B \subseteq F$ and $\tw(G[F]) <t$.
    Then in time $(|\mathcal{F}| \cdot n)^{\Oh(1)}$ one can compute a family $\mathcal{F'} \subseteq \mathcal{F}$ of size $\ell^{\Oh(\ell)}$ with the following promise: for every $F \in \mathcal{F}$ there exists $F' \in \mathcal{F'}$ such that \[(G[F],B) \equiv_{p,q} (G[F'],B)\qquad\textrm{and}\qquad \wei(F') \geq \wei(F).\]
\end{lemma}
\begin{proof}
    For every $F \in \mathcal{F}$, we compute a signature $\sigma_F$ of $(G[F],B)$ using \cref{lem:compute-signature}. Let $\Sigma\coloneqq \{\sigma_F\colon F\in \mathcal{F}\}$ be the set of all the obtained signatures. By \cref{lem:number-signature}, we must have $|\Sigma|\leq \ell^{\Oh(\ell)}$. For each $\sigma\in \Sigma$ we select a set $F_\sigma\in \mathcal{F}$ that has the largest weight among the sets $F\in \cal F$ with $\sigma_F=\sigma$. We define $\mathcal{F'}\coloneqq \{F_\sigma\colon \sigma\in \Sigma\}$; thus $|\mathcal{F'}|=|\Sigma|\leq \ell^{\Oh(\ell)}$. To see the that $\mathcal{F'}$ satisfies the promised property, observe that for each $F\in \mathcal{F}$, the set $F'\coloneqq F_{\sigma_F}$ belongs to~$\mathcal{F'}$, satisfies $\sigma_{F'}=\sigma_F$, which implies that $(G[F],B)\equiv_{p,q} (G[F'],B)$ by \cref{lem:equiv-signature}, and satisfies $\wei(F')\geq \wei(F)$ by the choice of $F'$.   
\end{proof}


\section{The algorithm}

In this section we prove \cref{thm:cmsotech}.

Let $\psi$ be a fixed \cmsotwo formula and $t$ be a constant.
Define \[\psi'\coloneqq \psi\land \phi_{\tw <t},\] where $\phi_{\tw <t}$ is the sentence provided by \cref{lem:twt}.
In other words, $\psi'$ verifies that the graph in question both satisfies $\psi$ and is of treewidth less than $t$.
Let $q$ be the quantifier rank of~$\psi'$ and let $p$ be the largest modulus used in any modular counting predicate used in $\psi'$.
Note that~$\psi'$, and thus also $q$ and $p$, are determined by $\psi$ and $t$; that is, they are fixed constants.

Let $G$ be a graph and let $\wei \colon V(G) \to \Q_+$ be a weight function on the vertices of $G$. Our goal is to find $F \subseteq V(G)$ such that $G[F] \models \psi'$ and $F$ is of maximum possible weight, or decide that no such $F$ exists.
Assume that $G$ is given with its tree decomposition $\cT = (T,\{X_u\}_{u \in V(T)})$ such that the tree-independence number of $\cT$ is at most $k$. Without loss of generality we may assume that $\cT$ is a nice tree decomposition.

For a node $u \in V(T)$ and a subset $B\subseteq X_u$,
we say that $F \subseteq V_u$ is \emph{feasible} for $(u,B)$
if $F \cap X_u = B$ and $\tw(G[F]) <t$.
For two sets $F,F' \subseteq V(G)$ that are feasible for $(u,B)$, we say that $F$ is a \emph{$B$-representative} for $F'$ (or that $F$ \emph{$B$-represents} $F'$) if \[(G[F],B)\equiv_{p,q}(G[F'],B)\qquad\textrm{and}\qquad \wei(F) \geq \wei(F').\]

Fix
\[\ell\coloneqq kt=\Oh(k).\]
For every node $u \in V(T)$ and every set $B \subseteq X_u$ of size at most $\ell$, we will compute a family $\mathsf{F}[u,B]$ of subsets of $V_u$ with the following properties:
\begin{enumerate}
    \item[(F1)] $|\mathsf{F}[u,B]| = \ell^{\Oh(\ell)} = k^{\Oh(k)}$;
    \item[(F2)] each $F \in \mathsf{F}[t,B]$ is feasible for $(u,B)$; and    
    \item[(F3)] for each set $F'$ that is feasible for $(u,B)$, there exists $F \in \mathsf{F}[u,B]$ that $B$-represents $F'$.
\end{enumerate}

Let us present the intuition behind this family.
Let $F_{\textsf{opt}}$ be an (unknown) optimum solution -- assuming that it exists.
Note that $\cT$, restricted to $F_{\textsf{opt}}$, is a tree decomposition of $G[F_{\textsf{opt}}]$.
As $\tw(G[F_{\textsf{opt}}]) <t$ implies in particular that $\chi(G[F_{\textsf{opt}}]) <t$, we conclude that each bag of $\cT$ contains at most $\ell = kt$ vertices from $F_{\textsf{opt}}$.
For each node $u \in V(T)$, we consider every possible intersection $B$ of $F_{\textsf{opt}}$ with $X_u$, and we aim to understand all possible ``behaviors'' of $G[F_{\textsf{opt}} \cap V_u]$ with $B$ as the boundary. These ``behaviors'' are formalized with rank-$(p,q)$ types.
As any two partial solutions that are equivalent satisfy the same formulae, it is sufficient to keep one -- the one with maximum weight.

\subsection{Obtaining the solution}
Suppose that we have computed the families $\mathsf{F}[\cdot,\cdot]$ as specified above.
Let $r$ be the root of $\cT$, that is, $V_r = V(G)$.
Since $\cT$ is a nice tree decomposition, we have $X_r = \emptyset$.

This means that $\mathcal{F} = \mathsf{F}[r,\emptyset]$ contains a maximum-weight representative of each equivalence class (and possibly some more sets, as distinct sets in $\mathcal{F}$ might still be equivalent). We are still left with deciding which of these sets (or, more precisely, which of subgraphs induced by one of these sets) satisfy $\psi'$.

This is, however, easy, as for each $F \in \mathcal{F}$ we have $\tw(G[F]) <t$. Consequently, we can apply \cref{thm:treewidthfpt} and then \cref{thm:courcelleb} to $G[F]$, for each $F \in \mathcal{F}$.
If such $G[F] \models \psi'$, we mark $F$ as a candidate, otherwise we discard it.

Finally, if any of the sets $F \in \mathcal{F}$ was marked as a candidate, we return one with the largest weight as a solution.
Otherwise, we report that no solution exists.
The total running time of this step is $\ell^{\Oh(\ell)} \cdot \Oh(n) = k^{\Oh(k)} \cdot \Oh(n)$.

The correctness of the procedure is straightforward.
On one hand, we never return a set that does not satisfy $\psi'$.
On the other hand, let $F_\textsf{opt}$ be an optimum solution.
By~(F3), there is some $F \in \mathcal{F}$ that has at least the same weight as $F_\textsf{opt}$ and is $(p,q)$-equivalent to $F_\textsf{opt}$ (with empty boundary).
In particular, $F$ also satisfies $\psi'$ so is also an optimum solution that is marked as a candidate in the procedure above.

\subsection{Computing families $\textsf{F}[\cdot,\cdot]$}

We now discuss how the families $\textsf{F}[\cdot,\cdot]$ can be computed by a bottom-up dynamic programming over $\cal T$.
The computation depends on the type of a node.
\paragraph{Leaf node.} If $u \in V(T)$ is a leaf node, we have $X_u = \emptyset$ and thus $V_u = \emptyset$.
Consequently, we set $\textsf{F}[u,\emptyset] = \{ \emptyset \}$.

\paragraph{Introduce node.} Let $u$ be an introduce node, $u'$ be its unique child in $T$, and let $\{v\} = X_u \setminus X_{u'}$.
Suppose we have computed all the suitable families $\textsf{F}[u',\cdot]$.
Now consider some $B \subseteq X_u$ of size at most $\ell$.

If $v \notin B$, we set $\textsf{F}[u,B] \coloneqq \textsf{F}[u',B]$.
Suppose now that $v \in B$.
For every $F \in \textsf{F}[u',B \setminus \{v\}]$ we proceed as follows.
We check whether $\tw(G[F \cup \{v\}]) <t$ using the algorithm of \cref{thm:treewidthfpt}.
If so, we include $F \cup \{v\}$ into $\textsf{F}[u,B]$, and otherwise we do nothing for this particular $F$.

\paragraph{Forget node.}
Let $u$ be a forget node, $u'$ be its unique child in $T$, and let $\{v\} = X_u' \setminus X_{u}$.
Suppose we have computed all the suitable families $\textsf{F}[u',\cdot]$.
Now consider some $B \subseteq X_u$ of size at most $\ell$.
We set  $\textsf{F}[u,B]$ to be the family obtained by applying \cref{lem:compress} to $\textsf{F}[u',B] \cup \textsf{F}[u', B \cup \{v\}]$.

\paragraph{Join node.}
Let $u$ be a join node and let $u',u''$ be its two children in $T$.
Suppose we have computed all the suitable families $\textsf{F}[u',\cdot]$ and $\textsf{F}[u'',\cdot]$.
Consider some $B \subseteq X_u$ of size at most $\ell$.

We set  $\textsf{F}[u,B]$ to be the family obtained by applying \cref{lem:compress} to the family 
\[\{F' \cup F'' ~|~ F' \in \textsf{F}[u',B], F'' \in \textsf{F}[u', B], \text{ and } \tw(G[F' \cup F'')] <t\}.\]

\subsection{Correctness}

We will argue about the correctness of each type of a node separately.

\paragraph{Leaf node.} If $u$ is a leaf node, $\emptyset$ is the only subset of $V_u$, so the correctness is straightforward.

\paragraph{Introduce node.} Consider an introduce node $u$,
let $u'$ be its only child and let $v$ be the unique vertex in $X_u \setminus X_{u'}$. Assume that the families $\textsf{F}[u',\cdot]$ are computed correctly.

If $v \notin B$, the correctness is straightforward.
Suppose now that $v \in B$.
As $|\textsf{F}[u,B]| \leq |\textsf{F}[u',B \setminus \{v\}]|$, condition (F1) clearly holds.
Condition (F2) holds too, as in $\textsf{F}[u,B]$ we include only sets feasible for $(u,B)$.
For (F3), consider any $\widetilde{F} \subseteq V_u$, feasible for $(u,B)$.
Note that $\widetilde{F}' \coloneqq \widetilde{F}  \setminus \{v\}$ is feasible for $(u',B \setminus \{v\})$.
Thus, there exists $F' \in \textsf{F}[u',B \setminus \{v\}]$,
such that the rank-$(p,q)$ types of $(G[\widetilde{F}'],B \setminus \{v\})$ and $(G[F'],B \setminus \{v\})$ are equal and $\wei(F') \geq \wei(\widetilde{F}')$. 

We observe that the neighborhood of $v$ in $V_u$ is contained in $X_u$. Since both $\widetilde{F}'$ and $F'$ are feasible for $(u',B\setminus \{v\})$, we in particular have 
\[N(v)\cap F'=N(v)\cap \widetilde{F}'= N(v)\cap (B \setminus \{v\}).\]
Thus, by \cref{prop:mso-type} we have
\begin{align*}
    \msotype^{p,q}(G[F' \cup \{v\}],B) = & \msotype^{p,q}(G[F'],B \setminus \{v\}) \oplus \msotype^{p,q}(G[B],B)\\
    \msotype^{p,q}(G[\widetilde{F}' \cup \{v\}],B) = & \msotype^{p,q}(G[\widetilde{F}'],B \setminus \{v\}) \oplus \msotype^{p,q}(G[B],B),
\end{align*}
so from the equivalence $(G[\widetilde{F}'],B \setminus \{v\})\equiv_{p,q} (G[F'],B \setminus \{v\})$ we conclude that \[\msotype^{p,q}(G[F' \cup \{v\}],B) = \msotype^{p,q}(G[\widetilde{F}' \cup \{v\}],B).\]
Note that $\phi_{\tw<t}$ is a subformula of $\psi'$, hence it is a \cpmsotwo formula of quantifier rank at most~$q$.
Therefore, the equality of types above implies that we have $G[F' \cup \{v\}]\models \phi_{\tw <t}$ if and only if $G[\widetilde{F}' \cup \{v\}]\models \phi_{\tw <t}$. Since the latter is true due to $\widetilde{F}=\widetilde{F}' \cup \{v\}$ being feasible for $(u,B)$, we conclude that also $\tw(G[F' \cup \{v\}<t$, and hence $F' \cup \{v\}$ gets included in $\textsf{F}[u,B]$.

Finally, since
\[
    \wei(F' \cup \{v\}) = \wei(F') + \wei(v) \geq \wei(\widetilde{F}') + \wei(v) = \wei(\widetilde{F}),
\]
property (F3) holds.

\paragraph{Forget node.}
Consider a forget node $u$, let $u'$ be its only child, and let $v$ be the unique vertex in $X_{u'} \setminus X_{u}$.
Assume that the families $\textsf{F}[u',\cdot]$ are computed correctly.

Fix any $B \subseteq X_u$ of size at most $\ell$, note that $v \notin B$.
Let $\mathcal{F}' \coloneqq \textsf{F}[u',B] \cup \textsf{F}[u',B \cup \{v\}]$.
As $\textsf{F}[u,B]$ is the result of applying \cref{lem:compress} to $\mathcal{F}'$, property (F1) holds.
As $\textsf{F}[u,B] \subseteq \mathcal{F}'$ and every set in $\mathcal{F}'$ is feasible for $(u,B)$, property (F2) holds, too.

For (F3), consider any $\widetilde{F} \subseteq V_u$, feasible for $(u,B)$. By \cref{lem:compress}, it is sufficient that we find an appropriate representative for $\widetilde{F}'$ in $\mathcal{F}'$.

Suppose first that $v \notin \widetilde{F}$.
Then $\widetilde{F}$ is feasible for $(u',B)$ and thus, there is $F' \in \textsf{F}[u',B]$ that $B$-represents $\widetilde{F}$.
This $F'$ was included in $\mathcal{F}'$.

Now suppose that $v \in \widetilde{F}$.
Then $\widetilde{F}$ is feasible for $(u',B \cup \{v\})$.
By an analogous argument as in the previous paragraph, we conclude that a $B$-representative for $\widetilde{F}$ was included in $\mathcal{F}'$.

\paragraph{Join node.}
Let $u$ be a join node and let $u'$ and $u''$ be its two children.
Suppose that the families $\textsf{F}[u',\cdot]$ and $\textsf{F}[u'',\cdot]$ are computed correctly.
Consider some $B \subseteq X_u$ and a set $\widetilde{F} \subseteq V_u$ that is feasible for $(u,B)$.
As $\textsf{F}[u,B]$ is the family obtained by applying \cref{lem:compress} to the family 
\[\mathcal{F} = \{F' \cup F'' ~|~ F' \in \textsf{F}[u',B], F'' \in \textsf{F}[u', B] \text{ and } \tw(G[F' \cup F'')] <t\},\]
it is sufficient to argue that there is a $B$-representative for $F$ in $\mathcal{F}$.
Let $\widetilde{F}' = \widetilde{F} \cap V_{u'}$ and $\widetilde{F}'' = \widetilde{F} \cap V_{u''}$.
Note that $\widetilde{F}' \cap \widetilde{F}'' = B$,
and thus, $\widetilde{F}'$ is feasible for $(u',B)$ and $\widetilde{F}''$ is feasible for~$(u'',B)$.
Consequently, there is a $B$-representative $F'$ for $\widetilde{F}'$ in $\textsf{F}[u',B]$ and a $B$-representative $F''$ for $\widetilde{F}''$ in $\textsf{F}[u'',B]$.
Let $F = F' \cup F''$.

Observe that by \cref{prop:mso-type},
\begin{align*}
\msotype^{p,q}(\widetilde{F},B) = & \msotype^{p,q}(\widetilde{F}',B) \oplus \msotype^{p,q}(\widetilde{F}'',B),\\
\msotype^{p,q}(F,B) = & \msotype^{p,q}(F',B) \oplus \msotype^{p,q}(F'',B).
\end{align*}
As the types of the corresponding factors on the right hand sides are equal, we conclude that \[\msotype^{p,q}(\widetilde{F},B)  = \msotype^{p,q}(F,B).\]
However, we still need to argue that $\tw(G[F]) <t$, so that $F$ is included in $\textsf{F}[u,B]$.
For this, we use the same argument as for the introduce nodes: $\phi_{\tw <t}$ is a subformula of $\psi'$, and hence it is a \cpmsotwo formula of quantifier rank at most~$q$.
Consequently, $G[F] \models \phi_{\tw <t}$ if and only if $G[\widetilde{F}] \models \phi_{\tw <t}$, and latter indeed holds as $\widetilde{F}$ is feasible for $u,B$.
Now, (F3) follows easily from~\cref{lem:compress}.

\subsection{Running time}
Recall that $V(T)$ can be assumed to be polynomial in $n$, and for each $u$, the number of choices for $B \subseteq X_u$ is at most $n^{\ell} = n^{\Oh(k)}$.

Let us now estimate the cost of computing $\textsf{F}[u,B]$ for a join node $u$, as this operation is the most involved.
This requires iterating over all elements of $\textsf{F}[u',B] \times \textsf{F}[u'',B]$, where $u',u''$ are the children of~$u$.
For each pair of sets, we call the algorithm from \cref{thm:treewidthfpt} for the graph induced by their union,
and finally we apply \cref{lem:compress} to the family of these sets that induce a graph of treewidth at most $t$.
Thus, each entry is computed in time $k^{\Oh(k)}$. The analysis for the other node types proceeds in a similar way.

As $k \leq n$, the overall complexity of the algorithm is $n^{\Oh(k)}$. This completes the proof of \cref{thm:cmsotech}.

\section{Conclusion}

In the work that first introduced the concept of the tree-independence number, Yolov~\cite{DBLP:conf/soda/Yolov18} pointed out a certain drawback of this parameter: it is unbounded for some classes of graphs that are ``simple in structure,'' like complete bipartite graphs.
As a remedy, he suggested a refined version of the parameter, which later became known as the \emph{induced matching treewidth}~\cite{DBLP:conf/esa/LimaMMORS24}.
Again, this parameter is defined using tree decompositions with certain structural properties: this time, we bound the size of a largest induced matching whose every edge intersects a single bag of the decomposition. Induced matching treewidth is bounded for all classes of bounded tree-independence number, but can still be small in classes allowing large complete bipartite graphs.

Yolov~\cite{DBLP:conf/soda/Yolov18} showed that \MWIS is polynomial-time-solvable in classes of bounded induced matching treewidth.
Lima et al.~\cite{DBLP:conf/esa/LimaMMORS24} showed that some other problems, notably \textsc{Feedback Vertex Set},
are also tractable in these graphs. Furthermore, they conjectured that an analogue of \cref{thm:cmso-old} holds for this parameter too, i.e., every problem expressible as $(\tw <t,\psi)$-\MWIS  can be solved in polynomial time for graphs with constant induced matching treewidth.
This conjecture was recently confirmed by Bodlaender, Fomin, and Korhonen~\cite{DBLP:journals/corr/abs-2507-07975}.

However, similarly to \cref{thm:cmso-old}, the running time of their algorithm is at least double exponential in the parameter bound.
We believe it would be interesting to see an analogue of our \cref{thm:cmsotech} for graphs of bounded induced matching treewidth.

\paragraph*{Acknowledgement.}
The authors are grateful to Maria Chudnovsky for persistent encouragement.

\bibliographystyle{abbrv}
\bibliography{main}
\end{document}

%% file: commands.tex
\usepackage[utf8]{inputenc}
\usepackage[T1]{fontenc}
\usepackage{amsfonts,amssymb,mathtools}
\usepackage{amsmath,amsthm,xcolor}
\usepackage{hyperref,cite,fullpage}
\usepackage{libertine}
\hypersetup{
	colorlinks=true,
    linkcolor={red!50!black},
    citecolor={blue!50!black},
    bookmarksopen=true,
	  bookmarksnumbered,
	  bookmarksopenlevel=2,
	  bookmarksdepth=3
}
\usepackage[capitalise]{cleveref}
\usepackage{xspace}
\usepackage{enumerate}
\usepackage{calc}
\usepackage{halloweenmath}
\usepackage{thmtools}
\usepackage{thm-restate}
\usepackage{tikz}
\usetikzlibrary{arrows.meta,patterns}

\definecolor{red}{rgb}{1,0,0}
\definecolor{blue}{rgb}{0,0,1}
\definecolor{green}{rgb}{0,1,0}
\definecolor{yellow}{rgb}{1,1,0}
\definecolor{orange}{rgb}{1,0.647,0}
\definecolor{gold}{rgb}{1,0.843,0}
\definecolor{purple}{rgb}{0.627,0.125,0.941}
\definecolor{gray}{rgb}{0.745,0.745,0.745}
\definecolor{brown}{rgb}{0.647,0.165,0.165}
\definecolor{navy}{rgb}{0,0,0.502}
\definecolor{pink}{rgb}{1,0.753,0.796}
\definecolor{seagreen}{rgb}{0.18,0.545,0.341}
\definecolor{turquoise}{rgb}{0.251,0.878,0.816}
\definecolor{violet}{rgb}{0.933,0.51,0.933}
\definecolor{darkblue}{rgb}{0,0,0.545}
\definecolor{darkcyan}{rgb}{0,0.545,0.545}
\definecolor{darkgray}{rgb}{0.663,0.663,0.663}
\definecolor{darkgreen}{rgb}{0,0.392,0}
\definecolor{darkmagenta}{rgb}{0.545,0,0.545}
\definecolor{darkorange}{rgb}{1,0.549,0}
\definecolor{darkred}{rgb}{0.545,0,0}
\definecolor{lightblue}{rgb}{0.678,0.847,0.902}
\definecolor{lightcyan}{rgb}{0.878,1,1}
\definecolor{lightgray}{rgb}{0.827,0.827,0.827}
\definecolor{lightgreen}{rgb}{0.565,0.933,0.565}
\definecolor{lightyellow}{rgb}{1,1,0.878}
\definecolor{black}{rgb}{0,0,0}
\definecolor{white}{rgb}{1,1,1}

\usetikzlibrary{calc}
\usepackage{lineno}

\newcommand{\cT}{\mathcal{T}}

\newcommand{\tw}{\mathrm{tw}}

\newcommand{\N}{\mathbb{N}}
\newcommand{\Q}{\mathbb{Q}}
\newcommand{\Z}{\mathbb{Z}}

\newcommand{\MWIS}{\textsc{MWIS}\xspace}

\renewcommand{\leq}{\leqslant}
\renewcommand{\geq}{\geqslant}

\newcommand{\wei}{\mathfrak{w}}

\renewcommand{\O}{\mathcal{O}}
\newcommand{\Oh}{\O}

\renewcommand{\setminus}{-}

\newcommand{\forget}{\mathsf{forget}}
\newcommand{\tup}[1]{{\bar #1}}

\newcommand{\cmsotwo}{$\mathsf{CMSO}_2$\xspace}
\newcommand{\cpmsotwo}{$\mathsf{C}_{\leq p}\mathsf{MSO}_2$\xspace}
\newcommand{\Formulae}{\mathsf{Formulae}}
\newcommand{\msotypes}{\mathsf{Types}}
\newcommand{\msotype}{\mathsf{type}}

\newcommand{\im}{\mathrm{im}\ }
\newcommand{\id}{\mathrm{id}}

\newtheorem{theorem}{Theorem}[section]

\newtheorem{corollary}[theorem]{Corollary}

\newtheorem{lemma}[theorem]{Lemma}

\theoremstyle{definition}

\crefname{claim}{Claim}{Claims}

\usepackage{todonotes}

\usepackage{framed}
\usepackage{tabularx}

\newlength{\RoundedBoxWidth}
\newsavebox{\GrayRoundedBox}
\newenvironment{GrayBox}[1]%
   {\setlength{\RoundedBoxWidth}{.93\columnwidth}
    \def\boxheading{#1}
    \begin{lrbox}{\GrayRoundedBox}
       \begin{minipage}{\RoundedBoxWidth}}%
   {   \end{minipage}
    \end{lrbox}
    \begin{center}
    \begin{tikzpicture}%
       \node(Text)[draw=black!20,fill=white,rounded corners,inner sep=2ex,text width=\RoundedBoxWidth]
             {\usebox{\GrayRoundedBox}};
        \coordinate(x) at (current bounding box.north west);
        \node [draw=white,rectangle,inner sep=3pt,anchor=north west,fill=white]
        at ($(x)+(6pt,.75em)$) {\boxheading};
    \end{tikzpicture}
    \end{center}}

\newenvironment{defproblemx}[1]{\noindent\ignorespaces%
                                \FrameSep=6pt%
                                \parindent=0pt%
                \begin{GrayBox}{#1}%
                \begin{tabular*}{\columnwidth}{!{\extracolsep{\fill}}@{\hspace{.1em}} >{\itshape} p{1.5cm} p{0.86\columnwidth} @{}}%
            }{
                \end{tabular*}%
                \end{GrayBox}%
                \ignorespacesafterend
            }

\newcommand{\problemTask}[3]{%
  \begin{defproblemx}{#1}
    Input: & #2 \\
    Task: & #3
  \end{defproblemx}
}